\documentclass[a4paper,11pt]{article}
\usepackage{test}
\usepackage{paralist,xspace}
\usepackage[authoryear,round]{natbib}
\usepackage{authblk}
\usepackage{hyperref}
\usepackage{doi}
\usepackage{palatino,mathpazo}
\usepackage[T1]{fontenc}
\usepackage[margin=1in]{geometry}

\renewcommand{\E}{E}
\newcommand{\Ex}[2]{\E\left[ #1 \,\mid\, #2 \right]}
\newcommand{\BHA}{Bewley--Huggett--Aiyagari\xspace}
\newcommand{\as}{\ensuremath{a.s.}\xspace}

\title{An Impossibility Theorem for Wealth in Heterogeneous-agent Models with
Limited Heterogeneity\thanks{We thank Jess Benhabib, Dan Cao, Chris Carroll, Makoto Nirei,
Neng Wang, and seminar participants at University of Queensland and Sapporo
Summer Workshop on Monetary and Financial Economics for comments and
suggestions. We are especially grateful to Yuichiro Waki for improving one of
our earlier results as Lemma \ref{lem:BRRA}.}}

\author[1]{John Stachurski\thanks{Email: \href{mailto:john.stachurski@anu.edu.au}{john.stachurski@anu.edu.au}.}}
\author[2]{Alexis Akira Toda\thanks{Email: \href{mailto:atoda@ucsd.edu}{atoda@ucsd.edu}.}}
\affil[1]{Research School of Economics, Australian National University}
\affil[2]{Department of Economics, University of California San Diego}

\numberwithin{equation}{section}
\numberwithin{thm}{section}

\begin{document}
\maketitle

\begin{abstract} 
    It has been conjectured 
    that canonical \BHA heterogeneous-agent models cannot explain
    the joint distribution of income and wealth. The results stated below verify this conjecture and
    clarify its implications under very general conditions.  We show in particular
    that if (i) agents are infinitely-lived, (ii) saving is risk-free, and
    (iii) agents have constant discount factors, then the wealth distribution
    inherits the tail behavior of income shocks (\eg, light-tailedness
    or the Pareto exponent).   Our restrictions on utility require only that
    relative risk aversion is bounded, and a large variety of income processes
    are admitted.  Our results show conclusively that it is necessary to go beyond
    standard models to explain the empirical fact that wealth is
    heavier-tailed than income. We demonstrate through examples that relaxing any of
    the above three conditions can generate Pareto tails.
\medskip


\medskip

\end{abstract}

\section{Introduction}

When studying wealth inequality, one empirical feature stands out as striking
and persistent over time and space: the wealth distribution exhibits a power
law tail. This fact was first discovered by
\cite*{Pareto1896LaCourbe,Pareto1897Cours} and has since been confirmed by
many studies.\footnote{See \cite*{gabaix2009,Gabaix2016JEP} for introductions
to power laws in economics.} A closely related observation is that the income
distribution is also heavy-tailed, although its Pareto exponent is
significantly larger, implying a heavier tail for wealth than
income.\footnote{\label{fn:exponent}The Pareto exponent for wealth is about
1.5 \citep*{Pareto1897Cours,KlassBihamLevyMalcaiSolomon2006,Vermeulen2018},
versus 2--3 for income \citep*{Atkinson2003,nirei-souma2007,Toda2012JEBO}.
Since the tail probability of a Pareto random variable satisfies $P(X>x)\sim
x^{-\alpha}$, where $\alpha>0$ is the Pareto exponent, a smaller value for
$\alpha$ corresponds to higher tail probability, implying a heavier tail (more
inequality).}


It is well known in the quantitative macroeconomics literature that canonical
\cite*{bewley1977,bewley1983,bewley1986}--\cite*{huggett1993}--\cite*{aiyagari1994}
models have difficulty in explaining the joint distribution of income and
wealth. For example, \cite*{aiyagari1994} documents that the wealth
Gini coefficient is 0.32 in the model, while it is 0.8 in the data.
\cite*{huggett1996} notes that the model-implied top 1\% wealth share is half
of that in the data. 
\cite*{KruegerMitmanPerri2016} argue that idiosyncratic unemployment risk and
incomplete financial markets alone cannot generate a sufficiently
dispersed wealth distribution,\footnote{The literature has extended the
``canonical'' \BHA model to explain the heavy-tailed wealth distribution by
introducing random discount factors \citep*{krusell-smith1998} and
idiosyncratic investment risk \citep*{benhabib-bisin-zhu2011}. We discuss this
literature in Section \ref{sec:possibility}.} even though such dispersion is
crucial for the study of aggregate fluctuations.
More specifically, \cite*{BenhabibBisinLuo2017} show that, in a
setting where income has a Pareto tail and agents use a linear
consumption rule, the Pareto exponent of wealth is either entirely
determined by the distribution of returns on wealth, or equal to the Pareto
exponent of income. They argue that similar results must obtain
with rational agents with constant relative risk aversion (CRRA) preferences
because, in such settings, the policy rules are asymptotically linear.\footnote{In a similar
    vein, \cite*{wang2007} shows that the constant absolute risk aversion
    (CARA) utility augmented with the \cite*{Uzawa1968} discounting function
    generates a stationary wealth distribution and that the wealth
    distribution cannot be more skewed than income.}

In this paper we confirm and significantly extend this conjecture by showing
that, for canonical \BHA models, all attempts to explain the large skewness of
the wealth distribution are bound to fail. By the qualification ``canonical'',
we mean models in which \begin{inparaenum}[(i)] \item agents are
infinitely-lived, \item saving is risk-free, and \item agents have constant
discount factors.  \end{inparaenum} In our main result (Theorem
\ref{thm:impossible}), we prove that the equilibrium wealth distribution inherits the tail
behavior of income shocks in \emph{any} such model.  This is an impossibility
theorem in the following sense: the tail thickness of the model output
(wealth) cannot exceed that of the input (income). If income is light-tailed
(\eg, bounded, Gaussian,  exponential, etc.), so is wealth. If income is
heavy-tailed, so is wealth, but with the same Pareto exponent, contradicting
the empirical relationship between the income and wealth distributions stated
above.  Thus,  one cannot produce a model consistent with the data without
relaxing at least one of assumptions (i)--(iii).

Our findings can be understood via the following intuition. In infinite-horizon dynamic
general equilibrium models, the discount factor $\beta>0$ and the gross
risk-free rate $R>0$ must satisfy the ``impatience'' condition $\beta R<1$,
for otherwise individual consumption diverges to infinity according to results
in \cite*{ChamberlainWilson2000}, which violates market clearing. But under
this impatience condition, we show that rational agents consume 
than what is implied by the permanent income hypothesis $c(a)=(1-1/R)a$, more
than the interest income, and the accumulation equation for wealth $a_t$
becomes a ``contraction'' in the sense that
\begin{equation}
a_{t+1}\le \rho a_t+y_{t+1}\label{eq:accumulate}
\end{equation}
for large enough $a_t$, where $y_{t+1}$ is income and $\rho$ is some positive
constant strictly less than 1. This inequality implies that the income shocks
die out in the long run, and hence the wealth distribution inherits the tail
behavior of income shocks. To obtain \eqref{eq:accumulate}, we use the results
from \cite*{LiStachurski2014}, who show the validity of policy function
iteration for solving income fluctuation problems. With the bound
\eqref{eq:accumulate} in hand, we characterize the tail behavior of wealth
using the properties of the moment generating function and applying several
inequalities such as Markov, H\"older, and Minkowski. 

Relative to the work of \cite*{BenhabibBisinLuo2017} discussed above, our
contributions are as follows: First, we provide a complete proof of the
impossibility result stated above in the context of an equilibrium model with
rational, optimizing agents, thereby confirming their conjecture on optimizing
households with CRRA utility in a general equilibrium
setting.  Second, our results are established in a class of models where
relative risk aversion need not be constant.  We require only that relative
risk aversion is asymptotically bounded. This means that minor deviations from standard
utility functions cannot reverse our results.  Similarly, our income process
is required only to have a finite mean.  Third, we provide a complete
analytical framework on tail thickness that accommodates both light-tailed and
heavy-tailed distributions, and connect it to the joint distribution of income
and wealth.  In the sense that we handle all classes of shocks
and allow for nonstationary additive processes, our proofs extend what is
contained in the related mathematical literature, such as the work of
\cite*{Grey1994} on Pareto tails.\footnote{The extra generality means that our
results are less sharp in other directions, although not in a way that
impinges on our main findings.  See the remark on page~\pageref{r:gk} for more
details.} Moreover, our proofs are almost completely self contained, and hence
can be readily adapted to subsequent research on income and wealth
distributions that tackles extensions to our framework.  

To tie up loose ends, we also show that the conditions of the impossibility
theorem are tight.  In Section \ref{sec:possibility}, we show through
examples that relaxing any of the three assumptions behind our main theorem
(agents are infinitely-lived, saving is risk-free, and the discount factor is
constant) can generate Pareto-tailed wealth distributions.  In doing so, we
draw on existing literature as it pertains to this topic and also
provide a simple, exactly solved model that features heterogeneous discount
factors and a Pareto-tailed wealth distribution.\footnote{This theory connects
    to several applied studies, such as
    \cite*{CarrollSlacalekTokuokaWhite2017} and \cite*{McKay2017}, which use
    numerically solved heterogeneous-agent quantitative models with several
agent types and different discount factors to generate skewed wealth
distributions.}

\subsection{Other literature}

Our work provides a logical converse to the findings of the growing literature
that relies on idiosyncratic investment risk (as opposed to earnings risk) to
explain the Pareto tail behavior in the wealth distribution.


As mentioned in the introduction, it has been known at least since
\cite*{aiyagari1994} that canonical \BHA models have difficulty in matching
the empirical wealth distribution. While the vast majority of papers in this
literature are numerical, several authors have theoretically shown that the
wealth distribution is bounded under certain assumptions.
\citet*[Theorems 3.8, 3.9]{SchechtmanEscudero1977} show the boundedness of wealth when income
is independent and identically distributed (i.i.d.) with a bounded support and
the utility function exhibits asymptotically constant relative risk aversion
(CRRA). \citet*[Proposition 4]{Aiyagari1993WP} relaxes the assumption on the
utility function to bounded relative risk aversion (BRRA). \cite*{huggett1993}
proves the boundedness of wealth when utility is CRRA and income is a
two-state Markov chain with a certain monotonicity property. To prove the
existence of a stationary equilibrium in an Aiyagari economy, building on
results in \cite*{LiStachurski2014} as we do, \citet*[Proposition
4]{Acikgoz2018} proves that wealth is bounded when the income process is a
finite-state Markov chain and the absolute risk aversion coefficient converges
to 0 as agents get richer (which is a weaker condition than BRRA).
\citet*[Proposition 3]{AchdouHanLasryLionsMollWP} show a similar result in a
continuous-time setting under the BRRA assumption.

Our contribution relative to this literature is that we do away with the
boundedness assumption on income and prove, in a fully micro-founded general
equilibrium setting, that wealth inherits the tail behavior of income, whether
it be light-tailed or heavy-tailed. This is a significant contribution, for
showing the boundedness alone does not tell us much about the tail behavior
because any unbounded distributions (with potentially different tail
properties) can be approximated by bounded ones. 



The key to proving our main result is to show that under the impatience
condition $\beta R<1$, agents uniformly consume more than the interest income
(Proposition \ref{prop:PIH}), which implies the AR(1) upper bound
\eqref{eq:accumulate} as a consequence of individual optimization and
equilibrium considerations.  This component of our paper is related to
\cite*{CarrollKimball1996} and \cite*{Jensen2018}, who prove the concavity of
the consumption function when the utility function exhibits hyperbolic
absolute risk aversion (HARA). The concavity of consumption implies a linear
lower bound (though not necessarily as tight as $c(a)\ge ma$ with $m>1-1/R$),
which we exploit to obtain an AR(1) upper bound on wealth accumulation as in
\eqref{eq:accumulate}. In contrast to these papers, we obtain the linear lower
bound on consumption under BRRA, which is a much weaker condition than HARA.

Finally, our paper is related to \cite*{BenhabibBisinZhu2015}, who show that a
\BHA model with idiosyncratic investment risk can generate a Pareto-tailed
wealth distribution. 
To obtain their possibility result, they derive a bound on wealth accumulation similar to \eqref{eq:accumulate} by assuming that agents have CRRA utility, and that earnings and investment risks are mutually independent and i.i.d.\ over time. Our paper is different in that
\begin{inparaenum}[(i)]
\item we focus on the impossibility result in the absence of financial risk, and
\item we consider a less restrictive environment, requiring only bounded (rather than constant) relative risk aversion and minimal restrictions on nonfinancial income.
\end{inparaenum}

\section{Tail thickness via moment generating function}

In this section we define several notions of tail thickness of random
variables using the moment generating function.  To this end, recall that, for
a random variable $X$ defined on some probability space $(\Omega, \mathcal F, P)$,
the moment generating function of $X$ is defined at $s \in \mathbb R$ by
$M_X(s)=\E[\e^{sX}]\in (0,\infty]$. We define light- and heavy-tailed random
variables as follows:

\begin{defn}[Tail thickness]
We say that a random variable $X$ has a \emph{light} upper tail if
$M_X(s)=\E[\e^{sX}]<\infty$ for some $s>0$. Otherwise
we say that $X$ has a \emph{heavy} upper tail.
\end{defn}

\begin{rem}
One can justify this definition as follows. A random variable $X$ is commonly referred to as having a heavy (Pareto) upper tail if there exist constants $A,\alpha>0$ such that $P(X>x)\ge Ax^{-\alpha}$ for large enough $x$, where $\alpha$ is the Pareto exponent. Since for $y\ge 0$ we have $\e^y\ge y^n/n!$ for any $n$ by considering the Taylor expansion, for any $s,x>0$ by Markov's inequality we obtain $$\E[\e^{sX}]\ge \E[\e^{sX}1_{X>x}]\ge \E[\e^{sx}1_{X>x}]=\e^{sx}P(X>x)\ge \frac{(sx)^n}{n!}Ax^{-\alpha}.$$ Taking $n>\alpha$ and letting $x\to\infty$, we obtain $\E[\e^{sX}]=\infty$.
\end{rem}

The following lemma shows that the tail probability of a light-tailed random variable has an exponential upper bound.

\begin{lem}\label{lem:exptail}
If $X$ is a light-tailed random variable, then
\begin{equation}
\label{eq:expbound}
P(X>x)\le M_X(s)\e^{-sx}
\quad \text{for all $x\in \R$}.
\end{equation}
\end{lem}

\begin{proof}
This is immediate from Markov's inequality, since, for any $x$, we have
$$M_X(s)=\E[\e^{sX}]\ge \E[\e^{sX}1_{X>x}]\ge \E[\e^{sx}1_{X>x}]=\e^{sx}P(X>x).
\qedhere$$
\end{proof}

The moment generating function $M_X(s)$ is convex and $M_X(0)=1$. Therefore the set $\set{s\in \R|M_X(s)<\infty}$ is an interval containing 0, so
\begin{equation}
\lambda=\sup\set{s\ge 0|M_X(s)<\infty}\in [0,\infty]\label{eq:abscissa}
\end{equation}
is well-defined. Taking the logarithm of \eqref{eq:expbound} for $s\in (0,\lambda)$, dividing by $x>0$, and letting $x\to\infty$, we obtain
$$\limsup_{x\to\infty}\frac{1}{x}\log P(X>x)\le -s.$$
Therefore letting $s\uparrow \lambda$, it follows that
\begin{equation}
\limsup_{x\to\infty}\frac{1}{x}\log P(X>x)\le -\lambda.\label{eq:expdecay}
\end{equation}
Using the argument in \citet*[pp.~42--43, Theorem 2.4a]{Widder1941}, one can
easily show that the inequality \eqref{eq:expdecay} is actually an equality,
although this fact plays no role in the subsequent discussion. Motivated by
\eqref{eq:expdecay}, we call the number $\lambda$ in \eqref{eq:abscissa} the
\emph{exponential decay rate} of the random variable $X$.

Next we categorize heavy-tailed random variables. Since the logarithm of a
Pareto random variable is exponential, it is convenient to define the tail
thickness based on the logarithm. Let $X$ be a heavy-tailed random variable
and $X_+=X1_{X\ge 0}$ be its positive part. The moment generating function of
$\log X_+$ is
$$M_{\log X_+}(s)=\E[\e^{s\log X_+}]=\E[X_+^s],$$
which is a convex function that is finite at $s=0$.\footnote{By convention, we
set $0^0=1$. Therefore $\E[X_+^s]\in (0,\infty]$ is well-defined for $s\ge 0$
and $\E[X_+^0]=1$.} By the same argument as above,
\begin{equation}
\alpha=\sup\set{s\ge 0|\E[X_+^s]<\infty}\in [0,\infty]\label{eq:Pareto}
\end{equation}
is well-defined and we have the property
$$\limsup_{x\to\infty}\frac{\log P(X>x)}{\log x}=-\alpha.$$
We call $\alpha$ the \emph{polynomial decay rate} of the random variable $X$.

So far we have defined the tail thickness of a single random variable $X$, but
as we shall see below, it is convenient to define similar concepts for a class
of random variables. Let $\mathcal{T}$ be some nonempty set and $(X_t)_{t\in
\mathcal{T}}$ be a collection of random variables defined on a common
probability space $(\Omega,\mathcal{F},P)$. Then we say that $(X_t)_{t\in
\mathcal{T}}$ is \emph{uniformly light-tailed} if there exists $s>0$ such that
$$M_{\mathcal{T}}(s):=\sup_{t\in \mathcal{T}}\E[\e^{sX_t}]<\infty.$$
If $(X_t)_{t\in \mathcal{T}}$ is uniformly light-tailed, then it immediately follows from Lemma \ref{lem:exptail} that
$$\sup_{t\in \mathcal{T}}P(X_t>x)\le M_{\mathcal{T}}(s)\e^{-sx}$$
for all $x$. Therefore the tail probabilities of $(X_t)_{t\in \mathcal{T}}$
can be uniformly bounded by an exponential function. By taking the supremum
over such $s>0$, we can define the exponential decay rate $\lambda$ of
$(X_t)_{t\in \mathcal{T}}$. We can similarly define uniformly heavy-tailed random variables and their polynomial decay rate $\alpha$ in the obvious way.

The following theorem, which is used in the proof of our main result, shows that a stochastic process that has a certain contraction property inherits the tail behavior of the shocks.

\begin{thm}\label{thm:tailbound}
Let $\phi:\R_+\to\R_+$ be a function such that
\begin{inparaenum}[(i)]
\item $\phi$ is bounded on any bounded set, and
\item $\rho:=\limsup_{x\to\infty} \phi(x)/x<1$.
\end{inparaenum}
Let $X_0\ge 0$ be some real number and $\set{X_t,Y_t}_{t=1}^\infty$ be a nonnegative stochastic process such that
\begin{equation}
X_t\le \phi(X_{t-1})+Y_t\label{eq:ub1}
\end{equation}
for all $t\ge 1$. Then the following statements are true.
\begin{enumerate}
\item\label{item:tailbound1} If $\set{Y_t}_{t=1}^\infty$ has a compact support, then so does $\set{X_t}_{t=1}^\infty$.
\item\label{item:tailbound2} If $\set{Y_t}_{t=1}^\infty$ is uniformly light-tailed with exponential decay rate $\lambda$, then $\set{X_t}_{t=1}^\infty$ is uniformly light-tailed with exponential decay rate $\lambda'\ge (1-\rho)\lambda$.
\item\label{item:tailbound3} If $\sup_t\E[Y_t]<\infty$ and $\set{Y_t}_{t=1}^\infty$ is uniformly heavy-tailed with polynomial decay rate $\alpha$, then $\set{X_t}_{t=1}^\infty$ has a polynomial decay rate $\alpha'\ge \alpha$.
\end{enumerate}
\end{thm}

\begin{rem}
It could be the case that $\set{Y_t}_{t=1}^\infty$ is heavy-tailed but $\set{X_t}_{t=1}^\infty$ is light-tailed. An obvious example is $\phi(x)\equiv 0$ and $X_t\equiv 0$, in which case the polynomial decay rate of $X_t$ is $\alpha'=\infty$.
\end{rem}

\begin{rem}
The lower bounds on the tail exponents in Theorem \ref{thm:tailbound} are sharp. To see this, suppose that $\rho\in [0,1)$, $X_0=0$, $X_t=\rho X_{t-1}+Y_t$, and $\set{Y_t}_{t=1}^\infty$ is perfectly correlated, so $Y_t=Y_1$ for all $t$. By iteration, we obtain
$$X_t=Y_t+\rho Y_{t-1}+\dots+\rho^{t-1}Y_1=\frac{1-\rho^t}{1-\rho}Y_1.$$
Therefore $X_t\to X=\frac{1}{1-\rho}Y_1$ as $t\to\infty$ almost surely. Hence if $Y_1$ is exponentially-distributed with decay rate $\lambda$ (\ie, $Y_1$ has density $f(x)=\lambda \e^{-\lambda x}$ for $x\ge 0$), then $X$ is exponentially distributed with decay rate $(1-\rho)\lambda$. If $Y_1$ is Pareto-distributed with Pareto exponent $\alpha$ and minimum size 1 (\ie, $Y_1$ has density $f(x)=\alpha x^{-\alpha-1}$ for $x\ge 1$), then $X$ is Pareto-distributed with Pareto exponent $\alpha$ and minimum size $\frac{1}{1-\rho}$.
\end{rem}

\begin{rem}
    \label{r:gk}
\cite*{Grey1994} studies the \cite*{kesten1973} process
\begin{equation}
X_t=A_tX_{t-1}+Y_t\label{eq:Kesten}
\end{equation}
when $\set{A_t,Y_t}_{t=1}^\infty$ is i.i.d.\ and shows under some assumptions
that $X_t$ and $Y_t$ have the same Pareto exponent (implying the same
polynomial decay rate). \cite*{Ghosh2010} extend this result to the Markovian
case. If we set $\phi(x)=\rho x$ in \eqref{eq:ub1} and assume that $Y_t$ is
Markovian, then Theorem \ref{thm:tailbound}(\ref{item:tailbound3}) can be
strengthened by setting $A_t=\rho$ in \eqref{eq:Kesten}. On the other hand,
Theorem~\ref{thm:tailbound} allows us to 
\begin{inparaenum}[(i)]
\item treat the bounded, light-tailed, and heavy-tailed cases
    simultaneously using elementary arguments and
\item handle potential non-stationarity in $\set{Y_t}_{t=1}^\infty$.
\end{inparaenum}
\end{rem}

\begin{rem}
\cite*{Mirek2011} studies the nonlinear recursion $X_t=\psi(X_{t-1},\theta_t)$ when $\set{\theta_t}$ is i.i.d.\ and shows that $X_t$ is heavy-tailed when $\psi$ is asymptotically linear in the sense that $M(\theta)=\lim_{x\to\infty}\psi(x,\theta)/x$ exists and $\abs{\psi(x,\theta)-M(\theta)x}\le Y(\theta)$ for some $Y(\theta)$. Some key assumptions are
\begin{inparaenum}[(i)]
\item $M(\theta)>1$ with positive probability (ensuring random growth, see Assumption H5) and
\item the ``additive'' term $Y(\theta)$ does not dominate the ``multiplicative'' term $M(\theta)$ (see Assumption H7).
\end{inparaenum}
Since our focus is impossibility of heavier tails without multiplicative risk, the assumption $\limsup_{x\to\infty}\phi(x)/x<1$ in Theorem \ref{thm:tailbound} suffices. \cite*{BenhabibBisinZhu2015} apply the \cite*{Mirek2011} result to show that wealth is heavy-tailed.
\end{rem}

\section{Wealth accumulation and tail behavior}

In this section we show that the wealth accumulation equation satisfies the AR(1) upper bound \eqref{eq:accumulate} under the impatience condition $\beta R<1$ and other weak conditions, which allows us to prove that wealth inherits the tail behavior of income.

We consider the following income fluctuation problem \citep*{SchechtmanEscudero1977}:
\begin{subequations}\label{eq:IF}
\begin{align}
&\maximize &&\E_0 \sum_{t=0}^\infty \beta^t u(c_t)\\
&\st && a_{t+1}=R(a_t-c_t)+y_{t+1},\label{eq:IFbudget}\\
&&& 0\le c_t\le a_t,\label{eq:IFborrowing}
\end{align}
\end{subequations}
where $u:\R_+\to \set{-\infty}\cup \R$ is the utility function, $\beta>0$ is the discount factor, $R>0$ is the gross risk-free rate,
$y_t\ge 0$ is income, $a_t$ is financial wealth at the beginning of period $t$ including current income, and the initial wealth $a_0>0$ is given. The constraint \eqref{eq:IFborrowing} implies that consumption must be nonnegative and the agent cannot borrow. The no borrowing assumption is without loss of generality as discussed in \cite*{ChamberlainWilson2000} and \cite*{LiStachurski2014}.

We impose standard assumptions on the utility function.

\begin{asmp}\label{asmp:utility}
The utility function is twice continuously differentiable on $\R_{++}=(0,\infty)$ and satisfies $u'>0$, $u''<0$, $u'(0)=\infty$, and $u'(\infty)=0$.\footnote{These assumptions are stronger than necessary. It suffices that $u$ is continuously differentiable on $(0,\infty)$, $u'>0$, $u'$ is strictly decreasing, and $u'(\infty)=0$. In particular, we do not need the twice differentiability and the Inada condition $u'(0)=\infty$, although the proof becomes more involved.}
\end{asmp}

Regarding the income process, we introduce the following assumption.

\begin{asmp}\label{asmp:boundedincome}
The income process $\set{y_t}$ takes the form $y_t=y(z_t)$, where $\set{z_t}$ is a Markov process on some set $Z$ and $\sup_{z\in Z}\Ex{y(z_{t+1})}{z_t=z}<\infty$.
\end{asmp}

Assumption \ref{asmp:boundedincome} is relatively weak because we have not specified the state space $Z$. The income process can be very general: for example, it can accommodate nonstationary life-cycle features. The only important assumption is that income has a (uniformly) finite conditional mean, which is natural in a stationary general equilibrium environment.\footnote{We use the uniform finiteness $\sup_{z\in Z}\Ex{y(z_{t+1})}{z_t=z}<\infty$ to apply Theorem \ref{thm:tailbound}(\ref{item:tailbound3}). \cite*{LiStachurski2014} prove the existence of a solution to the income fluctuation problem similar to \eqref{eq:IF} under a weaker finiteness assumption (see their Assumption 2.4).}

The reason why we assume that income is a function of a Markov process (not necessarily that income itself is Markovian) is for generality. For example, in the empirical literature on earnings dynamics, it is common to assume that income has a persistent-transitory decomposition
$$\log y_t=\xi_t+\eta_t,$$
where $\xi_t$ and $\eta_t$ are the persistent and transitory components, respectively \citep*{EjrnaesBrowning2014}. In this case $y_t=\exp(\xi_t+\eta_t)$ is a function of a Markov process $z_t=(\xi_t,\eta_t)$, but $y_t$ itself may not be Markovian. More generally, Assumption \ref{asmp:boundedincome} holds if $z_t=(\xi_t,\eta_t)$, $\set{\xi_t}$ follows a finite-state Markov chain, the distribution of $\eta_t$ depends only on $\xi_t$, and $\Ex{y(z_{t+1})}{\xi_t}<\infty$.

Due to strict concavity, the solution to the income fluctuation problem \eqref{eq:IF} is unique, if it exists. The following proposition shows that a solution exists under the ``impatience'' condition $\beta R<1$ and that it can be computed by policy function iteration. This result is essentially due to \citet*{LiStachurski2014} (see Appendix \ref{sec:politer} for details).

\begin{prop}\label{prop:politer}
Suppose Assumptions \ref{asmp:utility} and \ref{asmp:boundedincome} hold and $\beta R<1$. Then there exists a unique consumption policy function $c(a,z)$ that solves the income fluctuation problem \eqref{eq:IF}. Furthermore, we have $0<c(a,z)\le a$, $c$ is increasing in $a$, and $c(a,z)$ can be computed by policy function iteration.\footnote{To be precise, $c(a,z)$ is the limit obtained by iterating the map $K:\mathcal{C}\to\mathcal{C}$ starting from any $c_0\in \mathcal{C}$, where $\mathcal{C}$ is the set of candidate policy functions defined in Appendix \ref{sec:politer} and $(Kc)(a,z)$ is the value $t\in (0,a]$ that satisfies the Euler equation \eqref{eq:coleman}.}
\end{prop}

\begin{proof}
Immediate from Lemmas \ref{lem:metric}--\ref{lem:contraction} and the Banach fixed point theorem.
\end{proof}

To obtain the bound \eqref{eq:accumulate} for the wealth accumulation process, we assume that the utility function exhibits asymptotically bounded relative risk aversion (BRRA).\footnote{This assumption is stronger than necessary. It suffices to assume condition \eqref{eq:consratio}, which is slightly weaker than BRRA (see Lemma \ref{lem:BRRA}). However, we maintain BRRA because it is more intuitive and weak enough for all practical purposes.}

\begin{asmp}\label{asmp:BRRA}
The utility function $u$ satisfies $\limsup_{x\to\infty} \gamma(x)<\infty$, where
\begin{equation}
\gamma(x):=-\frac{xu''(x)}{u'(x)}>0\label{eq:RRA}
\end{equation}
is the local relative risk aversion coefficient.
\end{asmp}

The widely used constant relative risk aversion (CRRA) utility clearly satisfies Assumption \ref{asmp:BRRA}. More generally, let $u$ be the hyperbolic absolute risk aversion (HARA) utility, so
$$-\frac{u''(x)}{u'(x)}=\frac{1}{ax+b}$$
for some $a,b\in\R$, where $x$ takes values such that $ax+b>0$. It is well-known that the general functional form of HARA utility is
\begin{equation}
u(x)=\begin{cases}
\frac{1}{a-1}(ax+b)^{1-1/a},&(a\neq 0,1)\\
-b\e^{-x/b},&(a=0, b>0)\\
\log(x+b),&(a=1)
\end{cases}\label{eq:HARA}
\end{equation}
up to an affine transformation. Since the relative risk aversion of $u$ is
$$\gamma(x)=-\frac{xu''(x)}{u'(x)}=\frac{x}{ax+b},$$
we obtain
$$\limsup_{x\to\infty} \gamma(x)=\frac{1}{a}<\infty$$
if $a>0$, so Assumption \ref{asmp:BRRA} holds. On the other hand, the constant absolute risk aversion (CARA) utility exhibits relative risk aversion $\gamma(x)=x/b\to\infty$ as $x\to\infty$, so it violates Assumption \ref{asmp:BRRA}.

\begin{rem}
\cite*{SchechtmanEscudero1977} assume the following ``asymptotic exponent'' assumption:
\begin{equation}
\exists \gamma=\lim_{x\to\infty}-\frac{\log u'(x)}{\log x}.\label{eq:asymexp}
\end{equation}
Under twice continuous differentiability, the condition \eqref{eq:asymexp} is stronger than BRRA. To see this, by l'H\^opital's rule we have
$$\gamma=\lim_{x\to\infty}-\frac{\log u'(x)}{\log x}=\lim_{x\to\infty}-\frac{(\log u'(x))'}{(\log x)'}=\lim_{x\to\infty}-\frac{xu''(x)}{u'(x)}=\lim_{x\to\infty}\gamma(x),$$
so \eqref{eq:asymexp} implies that $u$ is asymptotically CRRA (in particular, BRRA). \citet*[Proposition 4]{Aiyagari1993WP} and \citet*[Proposition 3]{AchdouHanLasryLionsMollWP} assume that $u$ is BRRA. On the other hand, BRRA is stronger than the assumption used in \cite*{Rabault2002}, which is that the absolute risk aversion coefficient approaches 0:
\begin{equation}
\lim_{x\to\infty}-\frac{u''(x)}{u'(x)}=0.\label{eq:ARA0}
\end{equation}
In fact, if $u$ is BRRA, then
$$\lim_{x\to\infty}-\frac{u''(x)}{u'(x)}=\lim_{x\to\infty}\frac{\gamma(x)}{x}=0$$
so \eqref{eq:ARA0} holds, but the converse is not true. (As a counterexample, take $u$ such that $-u''(x)/u'(x)=x^{-\nu}$ for some $\nu\in (0,1)$.)
\end{rem}

The following proposition shows that under the impatience condition $\beta R<1$, agents uniformly consume more than the interest income, which is related to the permanent income hypothesis.\footnote{Interestingly, \cite*{Wang2003} and \cite*{Toda2017JEDC} show that the permanent income hypothesis holds in general equilibrium when the utility function is CARA, which is ruled out by Assumption \ref{asmp:BRRA}.} We use this result to bound the wealth accumulation process from above.

\begin{prop}\label{prop:PIH}
Suppose Assumptions \ref{asmp:utility}--\ref{asmp:BRRA} hold and $1\le R<1/\beta$. Let 
$$\bar{\gamma}=\limsup_{x\to\infty}-\frac{xu''(x)}{u'(x)}\in [0,\infty)$$
be the asymptotic relative risk aversion coefficient and $c(a,z)$ be the optimal consumption rule for the income fluctuation problem established in Proposition \ref{prop:politer}. 
Then for all $m\in (1-1/R,1-\beta^{1/\bar{\gamma}}R^{1/\bar{\gamma}-1})$, there exists an asset level $A\ge 0$ such that, for all $a \ge A$ and $z \in Z$, we have
\begin{equation}
c(a,z)\ge ma.\label{eq:PIH}
\end{equation}
\end{prop}

The reason we need Assumption \ref{asmp:BRRA} is because we are working with arbitrary (potentially unbounded) income processes. Under Assumption \ref{asmp:utility}, $\beta R<1$, and \eqref{eq:ARA0}, \citet*[Lemma C.1]{Rabault2002} shows that the next period's wealth in the consumption-saving problem becomes infinitely smaller than the current wealth $a$ as $a\to\infty$, which implies that $a$ is bounded. \citet*[Proposition 2]{Acikgoz2018} proves a similar result by relaxing the limit in \eqref{eq:ARA0} to $\liminf$. Although such arguments are enough for obtaining an upper bound for wealth when income is bounded, for the unbounded case we need $m>1-1/R$, for which a stronger assumption such as BRRA is necessary. Proposition \ref{prop:PIH} not only tells us that we can take some such number $m$, but it also gives an explicit choice: any number $m\in (1-1/R,1-\beta^{1/\bar{\gamma}}R^{1/\bar{\gamma}-1})$ will do, where the bounds depend only on the discount factor $\beta$, the gross risk-free rate $R$, and the asymptotic relative risk aversion coefficient $\bar{\gamma}$.

Using Proposition \ref{prop:PIH}, we can show that the wealth dynamics arising from the income fluctuation problem has the contraction property shown in \eqref{eq:accumulate} under the impatience condition $\beta R<1$.

\begin{prop}\label{prop:IF}
Suppose Assumptions \ref{asmp:utility}--\ref{asmp:BRRA} hold and $\beta R<1$. Let $\set{a_t}$ be the wealth arising from the solution to the income fluctuation problem \eqref{eq:IF}. Then the contraction property \eqref{eq:accumulate} holds for sufficiently high wealth level, and consequently the following statements are true:
\begin{enumerate}
\item If $\set{y_t}$ is uniformly light-tailed, then so is $\set{a_t}$.
\item If $\set{y_t}$ is uniformly heavy-tailed with polynomial decay rate $\alpha$, then $\set{a_t}$ has polynomial decay rate $\alpha'\ge \alpha$.
\end{enumerate}
Furthermore, the coefficient $\rho\in (0,1)$ in \eqref{eq:accumulate} can be chosen as follows:
\begin{itemize}
\item If $R<1$, then $\rho=R$.
\item If $R\ge 1$, then $\rho$ is any number in $((\beta R)^{1/\bar{\gamma}},1)$, where $\bar{\gamma}$ is the asymptotic relative risk aversion coefficient defined in Proposition \ref{prop:PIH}.
\end{itemize}
\end{prop}

\begin{proof}
If $R<1$, by the budget constraint \eqref{eq:IFbudget} we obtain
\begin{equation}
a_{t+1}\le Ra_t+y_{t+1}=\rho a_t+y_{t+1},\label{eq:aub1}
\end{equation}
where $\rho=R<1$. Hence \eqref{eq:accumulate} holds. If $1\le R<1/\beta$, by the budget constraint and Proposition \ref{prop:PIH}, we can take any $m\in (1-1/R,1-\beta^{1/\bar{\gamma}}R^{1/\bar{\gamma}-1})$ and some $A\ge 0$ such that
\begin{equation}
a_{t+1}\le R(a_t-ma_t)+y_{t+1}=\rho a_t+y_{t+1}\label{eq:aub2}
\end{equation}
for $\rho=R(1-m)\in ((\beta R)^{1/\bar{\gamma}},1)$ whenever $a_t\ge A$. Once again, the bound in \eqref{eq:accumulate} holds.

By \eqref{eq:aub1} and \eqref{eq:aub2}, letting $\phi(x)=\max\set{\rho x, RA}$, we obtain
$$a_{t+1}\le \phi(a_t)+y_{t+1}.$$
Therefore the claims follow from Theorem \ref{thm:tailbound}.
\end{proof}

The intuition for this result is as follows. Under the impatience condition $\beta R<1$, by Proposition \ref{prop:PIH}, agents uniformly consume more than the interest income at high wealth level. Since agents draw down their assets, wealth behaves similarly to income.

We now apply Proposition \ref{prop:IF} to show that in canonical heterogeneous-agent models, the wealth distribution cannot have a heavier tail than income, which is our main result. We formally define a \BHA model as follows.
\begin{defn}
A \emph{\BHA model} is any dynamic general equilibrium model such that ex ante identical, infinitely-lived agents solve an income fluctuation problem \eqref{eq:IF}.
\end{defn}

\begin{rem}
By requiring that the gross risk-free rate $R$ is constant over time in \eqref{eq:IFbudget}, we are excluding models with aggregate shocks. Thus, we are focusing on a stationary environment at the aggregate level, although the individual income processes may be nonstationary according to Assumption \ref{asmp:boundedincome}. We conjecture that our results extend to models with aggregate risk, although it is beyond the scope of this paper.
\end{rem}

While the preceding argument assumes the impatience condition $\beta R<1$, in general equilibrium models this condition is necessarily satisfied. To prove this impatience condition, we introduce an additional assumption on the income process.

\begin{asmp}\label{asmp:CW}
Let $\set{y_t}_{t=0}^\infty$ be the income process. There is an $\epsilon>0$ such that for any $x\in \R$, we have
\begin{equation}
\Pr\left(x\le \sum_{s=0}^\infty \beta^sy_{t+s}\le x+\epsilon \,|\, z^t\right)<1-\epsilon\label{eq:Ugamma}
\end{equation}
for all $t\ge 0$ and history $z^t=(z_0,\dots,z_t)$.
\end{asmp}
Assumption \ref{asmp:CW} is identical to condition (U$\gamma$) in \cite*{ChamberlainWilson2000}. Note that because $y_t\ge 0$ by assumption, the discounted sum of future income $\sum_{s=0}^\infty \beta^sy_{t+s}$ exists in $[0,\infty]$. Condition \eqref{eq:Ugamma} says that the conditional distribution of this discounted sum is not concentrated on a small enough interval. This condition is relatively weak and it roughly says that the discounted sum of income is stochastic. It holds, for example, if income is stationary and stochastic (nondeterministic).

\begin{thm}[Impossibility of heavy/heavier tails]\label{thm:impossible}
Consider a \BHA model such that Assumptions \ref{asmp:utility}--\ref{asmp:BRRA} hold. Suppose that an equilibrium with a wealth distribution with a finite mean exists and let $R>0$ be the equilibrium gross risk-free rate. If $R\neq 1/\beta$, then the following statements are true:
\begin{enumerate}
\item If the income process is light-tailed, then so is the wealth distribution.
\item If the income process is heavy-tailed with polynomial decay rate $\alpha$, then the wealth distribution has a polynomial decay rate $\alpha'\ge \alpha$.
\end{enumerate}
If in addition Assumption \ref{asmp:CW} holds, the condition $R\neq 1/\beta$ can be dropped.
\end{thm}

\begin{rem}
Although we present Theorem \ref{thm:impossible} as if the economy consists of ex ante identical agents, the result trivially generalizes to models with multiple agent types with heterogeneous preferences and income processes as long as each type satisfies the assumptions (in particular, the infinite horizon setting).
\end{rem}

The proof is an immediate consequence of Proposition \ref{prop:IF} combined with the convergence results in \cite*{ChamberlainWilson2000}. Theorem \ref{thm:impossible} is valuable since it places few assumptions. The only important assumption is that an equilibrium exists, which gives us the impatience condition $\beta R<1$ to apply Proposition \ref{prop:IF}. This assumption is minimal, for it is vacuous to study the wealth distribution unless an equilibrium exists. Regarding the income shocks, persistence and/or stationarity are irrelevant.

Theorem \ref{thm:impossible} has two important implications on the wealth distribution.
\begin{inparaenum}[(i)]
\item It is impossible to generate heavy-tailed wealth distributions from light-tailed income shocks.
\item If the income shock has a Pareto tail with exponent $\alpha$, the wealth distribution can have a Pareto tail, but its tail exponent $\alpha'$ can never fall below that of income shocks. Noting that smaller tail exponent means heavier tail, it follows that the wealth distribution cannot have a heavier tail than income.
\end{inparaenum}

Below, we discuss several applications of Theorem \ref{thm:impossible}.

\begin{exmp}
\cite*{aiyagari1994} uses the CRRA utility and a finite-state Markov chain for income. Hence by Theorem \ref{thm:impossible}, the wealth distribution is light-tailed. (In fact, it is bounded by applying Theorem \ref{thm:tailbound}(\ref{item:tailbound1}).\footnote{The boundedness result for the case with finite-state Markov chain is already known from \citet*[Proposition 4]{Aiyagari1993WP} and \citet*[Proposition 4]{Acikgoz2018}. The case when the income process is unbounded but light-tailed is new.})
\end{exmp}

\begin{exmp}
In \cite*{Quadrini2000}, even though there is idiosyncratic investment risk (stochastic returns), agents are restricted to only three levels of investment $\set{k_1,k_2,k_3}$ (see p.~25). Therefore the only investment vehicle that allows for unbounded investment is the risk-free asset, and the budget constraint reduces to one with stochastic income only. Since utility is CRRA, the wealth distribution is light-tailed.
\end{exmp}

\begin{exmp}
In \cite*{CastanedaDiazGimenezRiosRull2003}, the utility function is additively separable between consumption and leisure and the consumption part is CRRA. Since shocks follow a finite-state Markov chain, the wealth distribution is light-tailed.
\end{exmp}

\begin{exmp}
The budget constraint in \cite*{CagettiDeNardi2006} is
$$a'=(1-\delta)k+\theta k^\nu-(1+r)(k-a)-c,$$
where $a$ is risk-free asset and $k$ is capital (see their Equation (4) on p.~846). Here $\nu\in (0,1)$ is a parameter, $r$ is the net interest rate, $\delta\in (0,1)$ is capital depreciation rate, and $\theta$ is a random variable for productivity that has bounded support. Although there is some restriction on $k$, by ignoring the constraint and maximizing, we can bound the right-hand side by
$(1+r)a+Y-c$, where $Y$ is some random variable with bounded support. Since utility is CRRA, the wealth distribution is light-tailed.
\end{exmp}


\section{Possibility results}\label{sec:possibility}

Our Theorem \ref{thm:impossible} states that in canonical \BHA models in which
\begin{inparaenum}[(i)]
\item\label{item:possible.1} agents are infinitely-lived,
\item\label{item:possible.2} agents have constant discount factors, and
\item\label{item:possible.3} the only financial asset is risk-free,
\end{inparaenum}
the wealth distribution necessarily inherits the tail property of income. Thus, it is necessary to go beyond standard models to explain the empirical fact that wealth is heavier-tailed than income. A natural question is whether relaxing any of these assumptions can generate heavy-tailed wealth distributions from light-tailed income shocks. The answer is yes.

First, consider relaxing condition \eqref{item:possible.3} (saving is risk-free). In this case the return to investment is stochastic, and it is well-known that the wealth distribution can be heavy-tailed. (See, for example, \cite*{nirei-souma2007}, \cite*{benhabib-bisin-zhu2011,BenhabibBisinZhu2015}, \cite*{Toda2014JET}, \cite*{CaoLuo2017}, and the references in \cite*{BenhabibBisin2018}.) 
Next, consider relaxing condition \eqref{item:possible.2} (constant discounting). \cite*{krusell-smith1998} have numerically shown that when agents have random discount factors (\ie, they are more patient in some states than in others), the wealth distribution can be more skewed than the income distribution. Recently, \cite*{toda-discount} has theoretically proved that the wealth distribution can have a Pareto tail under random discounting even if there is no income risk, although the numerical value of the Pareto exponent is highly sensitive to the calibration of the discount factor process. Finally, consider relaxing condition \eqref{item:possible.1} (infinite horizon). \cite*{CarrollSlacalekTokuokaWhite2017} and \cite*{McKay2017} numerically solve heterogeneous-agent quantitative models with several agent types with different discount factors to generate skewed wealth distributions.\footnote{To our knowledge, \cite*{Samwick1998} is the first paper that uses a model with constant but heterogeneous discount factors to explain the wealth distribution.} 
Below, we provide a simplified version of such heterogeneous discount factor models 
and theoretically show that its wealth distribution is Pareto-tailed.\footnote{This example is essentially just a discrete-time, micro-founded, general equilibrium version of \cite*{WoldWhittle1957}, so there is nothing surprising in the results. Nevertheless, we think that it has a pedagogical value since \cite*{CarrollSlacalekTokuokaWhite2017} and \cite*{McKay2017} do not theoretically characterize the wealth distribution and the assumptions in \cite*{WoldWhittle1957} are endogenously satisfied in equilibrium.}

Consider an infinite-horizon endowment economy consisting of several agent types indexed by $j=1,\dots,J$. Let $\pi_j\in (0,1)$ be the fraction of type $j$ agents, where $\sum_{j=1}^J\pi_j=1$. Type $j$ agents are born and die with probability $p_j\in (0,1)$ every period, are endowed with constant endowment $y_j>0$ every period, and have CRRA utility
\begin{equation}
\E_0\sum_{t=0}^\infty \tilde{\beta}_j^t\frac{c_t^{1-\gamma_j}}{1-\gamma_j},\label{eq:CRRA}
\end{equation}
where $\tilde{\beta}_j=\beta_j(1-p_j)$ is the effective discount factor ($\beta_j\in (0,1)$ is the discount factor) and $\gamma_j>0$ is the relative risk aversion coefficient. There is a risk-free asset in zero net supply, and let $R>0$ be the gross risk-free rate determined in equilibrium. We assume that there is a perfectly competitive annuity market and let $\tilde{R}_j=\frac{R}{1-p_j}$ be the effective gross risk-free rate faced by type $j$ agents. We can show that a type $j$ agent maximizes utility \eqref{eq:CRRA} subject to the budget constraint
\begin{equation}
w_{t+1}=\tilde{R}_j(w_t-c_t),\label{eq:budget}
\end{equation}
where $w_t>0$ is wealth (financial wealth plus the annuity value of future income; see \cite*{toda-discount} for a rigorous discussion). A stationary equilibrium consists of a gross risk-free rate $R>0$, optimal consumption rules, and wealth distributions such that
\begin{inparaenum}[(i)]
\item agents optimize,
\item the commodity and risk-free asset markets clear, and
\item the wealth distributions are invariant.
\end{inparaenum}
The following theorem shows that a stationary equilibrium always exists and the wealth distribution exhibits a Pareto tail if and only if discount factors are heterogeneous across agent types.

\begin{thm}\label{thm:Pareto}
A stationary equilibrium exists. Letting $R>0$ be the equilibrium gross risk-free rate, the following statements are true.
\begin{enumerate}
\item If $\set{\beta_j}_{j=1}^J$ take at least two distinct values, then $\beta_jR>1$ for at least one $j$ and the stationary wealth distribution has a Pareto upper tail with exponent
\begin{equation}
\alpha=\min_{j:\beta_jR>1} \left[-\gamma_j\frac{\log (1-p_j)}{\log (\beta_jR)}\right]>1.\label{eq:zeta}
\end{equation}
\item If $\beta_1=\dots=\beta_J=\beta$, then $R=1/\beta$ and the wealth distribution of each type is degenerate.
\end{enumerate}
\end{thm}

Theorem \ref{thm:Pareto} shows that neither idiosyncratic investment risk nor random discounting are necessary for Pareto tails. Random birth/death is sufficient, although discount factor heterogeneity is necessary in this case.

\section{Concluding remarks}

In this paper we rigorously prove under weak conditions that, in canonical \BHA models in which (i) agents are infinitely-lived, (ii) saving is risk-free, and (iii) agents have constant discount factors, the wealth distribution always inherits the tail behavior of income shocks. The key intuition is that 
\begin{inparaenum}[(i)]
\item equilibrium considerations (market clearing) combined with the convergence results in \cite*{ChamberlainWilson2000} require the ``impatience'' condition $\beta R<1$, but
\item under this condition agents draw down their assets (Proposition \ref{prop:PIH}) and income shocks die out in the long run (Theorem \ref{thm:tailbound}).
\end{inparaenum}

The impossibility of heavier-tailed wealth distributions in canonical models comes from the fact that individual wealth shrinks with probability 1 as in \eqref{eq:accumulate}. The literature has long considered mechanisms to break the tight link between income and wealth, including random birth/death \citep*{WoldWhittle1957}, random discount factors \citep*{krusell-smith1998}, and idiosyncratic investment risk \citep*{benhabib-bisin-zhu2011}. In all of these cases, individual wealth can grow with positive probability, which essentially makes the wealth accumulation a random growth model (which is known to generate Pareto tails). Which mechanism is more important is an empirical question.

\appendix

\section{Proofs}\label{sec:proof}
\subsection{Proof of Theorem \ref{thm:tailbound}}
We first show that we may assume $\phi(x)=\rho x$ without loss of generality.
To this end, take $\rho'\in (\rho,1)$. By assumption, $\rho=\limsup_{x\to\infty}
\phi(x)/x<1$, so we can choose $\bar{x}$ such that $\phi(x)\le \rho'x$
for $x\ge \bar{x}$. Since $\phi$ is bounded on bounded sets, we
can choose $M\ge 0$ such that $\phi(x)\le M$ for $x\in [0,\bar{x}]$.
Therefore $\phi(x)\le \max\set{M,\rho'x}\le \rho'x+M$
for any $x\ge 0$, so \eqref{eq:ub1} implies
$X_t\le \phi(X_{t-1})+Y_t\le \rho'X_{t-1}+M+Y_t$.
If we define $Y_t'=Y_t+M$, then \eqref{eq:ub1} holds for $\phi(x)=\rho'x$ and
$Y_t=Y_t'$. Since adding a constant $M$ to $Y_t$ does not change its tail
behavior (\eg, boundedness, exponential decay rate, polynomial decay rate),
setting $\phi(x)=\rho x$ in \eqref{eq:ub1} costs no generality.

Iterating on \eqref{eq:ub1} with $\phi(x)=\rho x$ yields
\begin{equation}
X_t\le Y_t+\rho Y_{t-1}+\dots+\rho^{t-1}Y_1+\rho^t X_0.\label{eq:ub2}
\end{equation}

\begin{case}[$\set{Y_t}_{t=1}^\infty$ has a compact support]
Take $Y\ge 0$ such that $Y_t\in [0,Y]$ for all $t$. Then it follows from \eqref{eq:ub2} that
$$X_t\le (1+\rho+\dots+\rho^{t-1})Y+\rho^tX_0=\frac{1-\rho^t}{1-\rho}Y+\rho^t X_0\le \frac{1}{1-\rho}Y+X_0,$$
so $\set{X_t}_{t=1}^\infty$ is bounded.
\end{case}
\begin{case}[$\set{Y_t}_{t=1}^\infty$ is uniformly light-tailed]
Let $\lambda>0$ be the exponential decay rate. By definition, for any $s\in [0,\lambda)$, we have
\begin{equation}
f(s):=\sup_t \E[\e^{sY_t}]<\infty.\label{eq:MYub}
\end{equation}
In general, for any random variables $Z_1,Z_2$ and $\theta\in (0,1)$, by H\"older's inequality we have
\begin{align*}
M_{(1-\theta)Z_1+\theta Z_2}(s)&=\E[\e^{s((1-\theta)Z_1+\theta Z_2)}]=\E\left[\left(\e^{sZ_1}\right)^{1-\theta}\left(\e^{sZ_2}\right)^{\theta}\right]\\
&\le \E[\e^{sZ_1}]^{1-\theta}\E[\e^{sZ_2}]^\theta=M_{Z_1}(s)^{1-\theta}M_{Z_2}(s)^\theta.
\end{align*}
Multiplying both sides of \eqref{eq:ub2} by $1-\rho>0$, we get
\begin{equation}
(1-\rho)X_t\le \sum_{k=0}^t \theta_kY_k,\label{eq:ub3}
\end{equation}
where $Y_0\equiv (1-\rho)X_0$, $\theta_0=\rho^t$, and $\theta_k=(1-\rho)\rho^{t-k}$ for $k\ge 1$. Noting that $\theta_k\ge 0$ for all $k$ and $\sum_{k=0}^t \theta_k=1$, multiplying \eqref{eq:ub3} by $s>0$, taking the exponential, taking the expectation, and applying H\"older's inequality, it follows that
\begin{align*}
\E[\e^{(1-\rho)sX_t}]&\le \prod_{k=0}^t\E[\e^{sY_k}]^{\theta_k}\\
&=\e^{(1-\rho)\rho^t sX_0}\prod_{k=1}^t\E[\e^{sY_k}]^{\theta_k}\le \e^{(1-\rho)\rho^t sX_0} f(s)^{1-\rho^t},
\end{align*}
where $f(s)$ is as in \eqref{eq:MYub}. Redefining $(1-\rho)s$ as $s$ and noting that $0<\rho^t<1$ and $X_0\ge 0$, we obtain
$$\E[\e^{sX_t}]\le \e^{sX_0}\max\set{1,f\left(\frac{s}{1-\rho}\right)}.$$
By the definition of the exponential decay rate $\lambda>0$, the right-hand side is finite if $\frac{s}{1-\rho}<\lambda\iff s<(1-\rho)\lambda$. Therefore by definition $\set{X_t}_{t=1}^\infty$ is uniformly light-tailed, and the exponential decay rate satisfies $\lambda'\ge (1-\rho)\lambda$.
\end{case}
\begin{case}[$\set{Y_t}_{t=1}^\infty$ is uniformly integrable and heavy-tailed]
Since by assumption $\sup_t\E[Y_t^1]=\sup_t\E[Y_t]<\infty$, the polynomial
decay rate satisfies $\alpha\ge 1$. Let $s\in [1,\alpha]$. Applying
Minkowski's inequality to both sides of \eqref{eq:ub1} yields
$\E[X_t^s]^{1/s}\le \rho \E[X_{t-1}^s]^{1/s}+\E[Y_t^s]^{1/s}$.
Letting $f(s)=\sup_t \E[Y_t^s]$ and iterating, we get
$$\E[X_t^s]^{1/s}\le \frac{1-\rho^t}{1-\rho}f(s)^{1/s}+\rho^tX_0\le \frac{1}{1-\rho}f(s)^{1/s}+X_0.$$
Therefore
$$\E[X_t^s]\le \left(\frac{1}{1-\rho}f(s)^{1/s}+X_0\right)^s.$$
Since the right-hand side does not depend on $t$ and $f(s)<\infty$ for $s=1$ and $s<\alpha$ by definition, it follows that the polynomial decay rate of $\set{X_t}_{t=1}^\infty$ satisfies $\alpha'\ge \alpha\ge 1$. \qedsymbol
\end{case}

\subsection{Proof of Proposition \ref{prop:PIH}}
To prove Proposition \ref{prop:PIH}, we first use the fact that the optimal consumption rule in the original problem can be bounded below by that with zero income.

\begin{lem}\label{lem:clb}
Suppose Assumptions \ref{asmp:utility} and \ref{asmp:boundedincome} hold and $\beta R<1$. Given asset $a>0$ and state $z$, let $c(a,z),c_0(a)$ be the optimal consumption rules for the income fluctuation problem \eqref{eq:IF} with and without income (which are established in Proposition \ref{prop:politer}). Then $c(a,z)\ge c_0(a)$.
\end{lem}

\begin{proof}
If an optimal consumption rule exists, by considering whether the no borrowing constraint $c\le a$ binds or not, it must satisfy the Euler equation
\begin{equation}
u'(c(a,z))=\max\set{\beta R\Ex{u'(c(R(a-c(a,z))+y',z'))}{z},u'(a)}.\label{eq:euler}
\end{equation}
We use policy function iteration as discussed in Appendix \ref{sec:politer} to characterize properties of $c(a,z)$.

Given a candidate policy function $c(a,z)$, define the policy function (Coleman) operator $(Kc)(a,z)$ as the unique value $0<t\le a$ that solves the equation
\begin{equation}
u'(t)=\max\set{\beta R\Ex{u'(c(R(a-t)+y',z'))}{z},u'(a)}.\label{eq:coleman}
\end{equation}
Lemma \ref{lem:selfmap} shows that $K$ is a well-defined monotone self map. Proposition \ref{prop:politer} shows that $K$ has a unique fixed point and $K^nc$ converges to this fixed point as $n\to\infty$. Because the optimal consumption policy is a fixed point of $K$, which is unique, it suffices to show $Kc_0(a)\ge c_0(a)$, for if that is the case we obtain $c_0(a)\le (K^nc_0)(a)\to c(a,z)$.

Let $t=(Kc_0)(a)$ solve \eqref{eq:coleman} for $c(a,z)=c_0(a)$. To show $t\ge c_0(a)$, suppose on the contrary that $t<c_0(a)$. Since by Proposition \ref{prop:politer} $c_0(a)$ is increasing in $a$, $y' \ge 0$, and $t<c_0(a)$, we have
$$c_0(R(a-t)+y')\ge c_0(R(a-c_0(a))).$$
Since $u'$ is strictly decreasing and $c_0(a)$ satisfies \eqref{eq:euler} for $y'\equiv 0$, we obtain
\begin{align*}
u'(t)>u'(c_0(a))&=\max\set{\beta R\Ex{u'(c_0(R(a-c_0(a))))}{z},u'(a)}\\
&\ge \max\set{\beta R\Ex{u'(c_0(R(a-t)+y'))}{z},u'(a)}=u'(t),
\end{align*}
which is a contradiction. Therefore $(Kc_0)(a)=t\ge c_0(a)$.
\end{proof}

Next, we derive a useful implication of BRRA. Consider the following condition:
\begin{equation}
\text{For any constant $\kappa\in (0,1)$, we have}\quad \liminf_{x\to\infty}\frac{(u')^{-1}(\kappa u'(x))}{x}>1.\label{eq:consratio}
\end{equation}
Condition \eqref{eq:consratio} is relatively weak. To see this, since $u'>0$, $\kappa\in (0,1)$, and $u''<0$ (hence $u'$ and $(u')^{-1}$ are decreasing), we always have
$$\frac{(u')^{-1}(\kappa u'(x))}{x}>\frac{(u')^{-1}(u'(x))}{x}=1.$$
Condition \eqref{eq:consratio} adds a degree of uniformity to this bound at infinity. The following lemma shows that bounded relative risk aversion (BRRA, Assumption \ref{asmp:BRRA}) is sufficient for condition \eqref{eq:consratio} to hold, and almost necessary.

\begin{lem}\label{lem:BRRA}
Let $\gamma(x)=-xu''(x)/u'(x)$ be the local relative risk aversion coefficient. Then the following statements are true.
\begin{enumerate}
\item If $\limsup_{x\to\infty}\gamma(x)<\infty$, then condition \eqref{eq:consratio} holds.
\item If $\lim_{x\to\infty}\gamma(x)=\infty$, then condition \eqref{eq:consratio} fails.
\end{enumerate}
\end{lem}

\begin{proof}
Take any $\kappa\in (0,1)$. For any $x>0$, define $y=(u')^{-1}(\kappa u'(x))$. By definition, $u'(y)/u'(x)=\kappa\in (0,1)$. Since $u''<0$, we have $y>x$. By the Fundamental Theorem of Calculus and \eqref{eq:RRA}, we obtain
\begin{align}
-\log \kappa&=\log u'(x)-\log u'(y)=-\int_1^{y/x}\frac{\partial}{\partial s}\log u'(xs)\diff s\notag \\
&=-\int_1^{y/x}\frac{xu''(xs)}{u'(xs)}\diff s=\int_1^{y/x}\frac{\gamma(xs)}{s}\diff s.\label{eq:logalpha}
\end{align}
\begin{enumerate}
\item If $\limsup_{x\to\infty}\gamma(x)<\infty$, then there exists $M<\infty$ such that $\gamma(x)\le M$ for large enough $x$. Then \eqref{eq:logalpha} implies
$$-\log\kappa\le \int_1^{y/x}\frac{M}{s}\diff s=M\log\frac{y}{x}\iff \frac{y}{x}\ge \kappa^{-1/M}.$$
Since this inequality holds for large enough $x$ and $y=(u')^{-1}(\kappa u'(x))$, we obtain
\begin{equation}
\liminf_{x\to\infty}\frac{(u')^{-1}(\kappa u'(x))}{x}\ge \kappa^{-1/M}>1,\label{eq:liminf1}
\end{equation}
which is \eqref{eq:consratio}.
\item If $\lim_{x\to\infty}\gamma(x)=\infty$, take any $M>0$ and choose $\bar{x}>0$ such that $\gamma(x)\ge M$ for all $x\ge \bar{x}$. Then for $x\ge \bar{x}$, by \eqref{eq:logalpha} we obtain
$$-\log \kappa\ge \int_1^{y/x}\frac{M}{s}\diff s=M\log\frac{y}{x}\iff \frac{y}{x}\le \kappa^{-1/M}.$$
Since this inequality holds for large enough $x$ and $y=(u')^{-1}(\kappa u'(x))$, we obtain
$$\liminf_{x\to\infty}\frac{(u')^{-1}(\kappa u'(x))}{x}\le \kappa^{-1/M}\to 1$$
as $M\to\infty$, so \eqref{eq:consratio} fails. \qedhere
\end{enumerate}
\end{proof}

We now use condition \eqref{eq:consratio} to prove Proposition \ref{prop:PIH}.

\begin{proof}[\bf Proof of Proposition \ref{prop:PIH}]
Let $K$ be the policy function operator defined in \eqref{eq:coleman} associated with the zero income model, so for a consumption policy $c=c(a)$, the number $t=(Kc)(a)$ solves
\begin{equation}
u'(t)=\max\set{\beta Ru'(c(R(a-t))),u'(a)}.\label{eq:foc}
\end{equation}
Fix some $m\in (1-1/R,1)$, $A>0$, and consider the candidate policy
\begin{equation}
c(a)=\begin{cases}
\epsilon c_0(a),&(0<a<A)\\
ma,&(a\ge A)
\end{cases}\label{eq:clb}
\end{equation}
where $\epsilon\in (0,1)$ is a number such that $\epsilon c_0(A)\le mA$. Clearly $c(a)$ in \eqref{eq:clb} is increasing and satisfies $0<c(a)\le a$. If we can show $(Kc)(a)\ge c(a)$ for all $a$, then by Proposition \ref{prop:politer} and Lemma \ref{lem:clb}, we obtain
$$c(a)\le (K^nc)(a)\to c_0(a)\le c(a,z),$$
which implies \eqref{eq:PIH} for $a\ge A$.

\begin{step}
Let $\bar{\gamma}=\limsup_{x\to\infty}-xu''(x)/u'(x)<\infty$ be the asymptotic relative risk aversion coefficient and define $\bar{m}$ by $(\beta R)^{1/\bar{\gamma}}=R(1-\bar{m})$. Then $\bar{m}\in (1-1/R,1]$. Furthermore, for any $m\in (1-1/R,\bar{m})$, we have
\begin{equation}
\liminf_{x\to\infty}\frac{(u')^{-1}(\beta R u(x))}{x}>\frac{1}{R(1-m)}.\label{eq:liminf2}
\end{equation}
\end{step}

Since by assumption $\beta R<1$ and $\bar{\gamma}\in [0,\infty)$, we have $(\beta R)^{1/\bar{\gamma}}\in [0,1)$. Since $(\beta R)^{1/\bar{\gamma}}=R(1-\bar{m})$ and $R\ge 1$ by assumption, we have $\bar{m}\in (1-1/R,1]$. Take any $m\in (1-1/R,\bar{m})$. Then $(\beta R)^{1/\bar{\gamma}}=R(1-\bar{m})<R(1-m)$. Since $\beta R<1$, we can take sufficiently small $M>\bar{\gamma}$ such that $(\beta R)^{1/M}<R(1-m)$. Letting $\kappa=\beta R<1$ in \eqref{eq:liminf1}, we obtain \eqref{eq:liminf2}.

\begin{step}
The following statement is true:
\begin{equation}
(\forall m\in (1-1/R,\bar{m}))(\exists A>0)(\forall a\ge A)(t=(Kc)(a)\ge ma).\label{eq:wanttoshow}
\end{equation}
\end{step}
In seeking a contradiction, suppose
\begin{equation}
(\exists m\in (1-1/R,\bar{m}))(\forall A>0)(\exists a\ge A)(t=(Kc)(a)<ma).\label{eq:negation}
\end{equation}
If $\beta Ru'(c(R(a-t)))<u'(a)$, by \eqref{eq:foc} we have $u'(t)=u'(a)\iff t=a\ge ma$, which contradicts \eqref{eq:negation}. Therefore $\beta Ru'(c(R(a-t)))\ge u'(a)$. Noting that we are considering the candidate policy \eqref{eq:clb}, it follows from \eqref{eq:foc} that
$$g(t):=u'(t)-\beta Ru'(mR(a-t))=0.$$
Since $g'(t)=u''(t)+\beta mR^2u''(mR(a-t))<0$, $g$ is strictly decreasing. Since $t<ma$ by \eqref{eq:negation}, we obtain
\begin{align*}
&0=g(t)>g(ma)=u'(ma)-\beta Ru'(mR(1-m)a)\\
\implies &\frac{(u')^{-1}(\beta R u'(mR(1-m)a))}{mR(1-m)a}<\frac{1}{R(1-m)}.
\end{align*}
By \eqref{eq:negation}, $a>0$ can be taken arbitrarily large. Therefore letting $a\to\infty$ and $x=mR(1-m)a$, we obtain
$$\liminf_{x\to\infty}\frac{(u')^{-1}(\beta R u'(x))}{x}\le \frac{1}{R(1-m)},$$
which contradicts \eqref{eq:liminf2}. Therefore \eqref{eq:wanttoshow} holds.

\begin{step}
The bound \eqref{eq:PIH} holds.
\end{step}
Take any $m\in (1-1/R,\bar{m})$ and $A>0$ such that \eqref{eq:wanttoshow} holds. Take $\epsilon\in (0,1)$ such that $\epsilon c_0(A)\le mA$ and define $c(a)$ as in \eqref{eq:clb}. By \eqref{eq:wanttoshow}, we have $(Kc)(a)\ge ma=c(a)$ for $a\ge A$. Therefore it remains to show $(Kc)(a)\ge c(a)$ for $a<A$. Suppose on the contrary that $(Kc)(a)<c(a)$ for some $a<A$. Since $c(a)=\epsilon c_0(a)$ and $\epsilon\in (0,1)$, we have $(Kc)(a)<c_0(a)$. Applying $K$ and using monotonicity and Proposition \ref{prop:politer}, it follows that
$$c_0(a)>(Kc)(a)\ge (K^nc)(a)\to c_0(a),$$
which is a contradiction. Therefore $(Kc)(a)\ge c(a)$. The rest of the proof follows from the discussion at the beginning of the proof.
\end{proof}

\subsection{Proof of Theorem \ref{thm:impossible}}

First let us show $\beta R\le 1$. By considering whether the borrowing constraint binds or not, the Euler equation becomes
$$u'(c_t)=\max\set{\beta R\E_t[u'(c_{t+1})],u'(a_t)}.$$
Therefore $u'(c_t)\ge \beta R\E_t[u'(c_{t+1})]$. Multiplying both sides by $(\beta R)^t>0$ and letting $M_t=(\beta R)^tu'(c_t)>0$, we obtain $M_t\ge \E_t[M_{t+1}]$. Since $M_0=u'(c_0)<\infty$, the process $\set{M_t}_{t=0}^\infty$ is a supermartingale. By the Martingale Convergence Theorem \cite*[p.~148]{Pollard2002}, there exists an integrable random variable $M\ge 0$ such that $\lim_{t\to\infty}M_t=M$ almost surely. Since $\E[M]<\infty$, we have $M<\infty$ (\as). If $\beta R>1$, then
$$u'(c_t)=(\beta R)^{-t}M_t\to 0\cdot M=0\quad (\as),$$
so $c_t\to\infty$ (\as) because $u'>0$ and $u'(\infty)=0$. (This argument is exactly the same as Theorem 2 of \cite*{ChamberlainWilson2000}.) Since this is true for any agent, the aggregate consumption diverges to infinity, which is impossible in a model with a wealth distribution with a finite mean. Therefore $\beta R\le 1$.

If $R\neq 1/\beta$, then $\beta R<1$, so the conclusion follows from Proposition \ref{prop:IF}.

Finally, suppose Assumption \ref{asmp:CW} holds. If $\beta R=1$, then by \citet*[Theorem 4]{ChamberlainWilson2000} we have $c_t\to\infty$ (\as), which is a contradiction. \qedsymbol

\subsection{Proof of Theorem \ref{thm:Pareto}}
Since the proof is similar to \citet*[Theorems 3, 4]{toda-discount}, we only provide a sketch.

It is well-known that the optimal consumption rule for maximizing the CRRA utility \eqref{eq:CRRA} subject to the budget constraint \eqref{eq:budget} is
$$c=\left(1-\tilde{\beta}_j^{1/\gamma_j}\tilde{R}_j^{1/\gamma_j-1}\right)w.$$
(See \cite*{levhari-srinivasan1969}, \cite*{samuelson1969}, or more generally, \cite*{Toda2014JET,toda-discount} for the solution in a Markovian environment.) Then the law of motion for wealth becomes
\begin{equation}
w_{t+1}=\left(\tilde{\beta}_j\tilde{R}_j\right)^{1/\gamma_j}w_t=(\beta_j R)^{1/\gamma_j}w_t.\label{eq:wdynamics}
\end{equation}
Let $W_j$ be the average wealth of type $j$ agents. Since agents are born and die with probability $p_j$, by accounting we obtain
\begin{equation}
W_j=(1-p_j)(\beta_j R)^{1/\gamma_j}W_j+p_jw_{j0},\label{eq:account}
\end{equation}
where
\begin{equation}
w_{j0}=\sum_{t=0}^\infty \tilde{R}_j^{-t}y_j=\frac{\tilde{R}_j}{\tilde{R}_j-1}y_j\label{eq:wj0}
\end{equation}
is the initial wealth (present discounted value of future income) of type $j$ agents. Assuming $(1-p_j)(\beta_j R)^{1/\gamma_j}<1$, we can solve \eqref{eq:account} to obtain
\begin{equation}
W_j=\frac{p_jw_{j0}}{1-(1-p_j)(\beta_j R)^{1/\gamma_j}}.\label{eq:Wj}
\end{equation}
Combining \eqref{eq:wj0} and \eqref{eq:Wj}, the market clearing condition becomes
\begin{equation}
0=\sum_{j=1}^J\pi_j(W_j-w_{j0})=\sum_{j=1}^J\frac{R\pi_jy_j\left((\beta_j R)^{1/\gamma_j}-1\right)}{\left(\frac{R}{1-p_j}-1\right)\left(1-(1-p_j)(\beta_j R)^{1/\gamma_j}\right)}.\label{eq:clear}
\end{equation}
Let $f(R)$ be the right-hand side of \eqref{eq:clear} and $\bar{R}=\min_j[\beta_j(1-p_j)^{\gamma_j}]^{-1}$, which is greater than 1 since $\beta_j,p_j\in (0,1)$ and $\gamma_j>0$. Since
$$\frac{R}{1-p_j}>1,\quad (1-p_j)(\beta_j R)^{1/\gamma_j}<1$$
on $R\in [1,\bar{R})$, the function $f(R)$ is well-defined in this range and the denominators are positive. Since $\beta_j<1$, we have $f(1)<0$. Since $(\beta_j R)^{1/\gamma_j}\to \frac{1}{1-p_j}>1$ as $R\uparrow \bar{R}$ for $j$ that achieves the minimum in the definition of $\bar{R}$, we have $f(R)\uparrow \infty$ as $R\uparrow \bar{R}$. Clearly $f$ is continuous on $[1,\bar{R})$. Therefore there exists $R\in (1,\bar{R})$ such that $f(R)=0$, so an equilibrium exists.

To show the implications for the wealth distribution, first consider the case $\beta_1=\dots=\beta_J=\beta$. Then
$$f(R)=\sum_{j=1}^J\frac{R\pi_jy_j\left((\beta R)^{1/\gamma_j}-1\right)}{\left(\frac{R}{1-p_j}-1\right)\left(1-(1-p_j)(\beta R)^{1/\gamma_j}\right)}.$$
Since the denominators are positive on $[1,\bar{R})$ and the sign of the numerators depends only on the magnitude of $\beta R$ relative to 1, we have $f(R)\gtrless 0$ according as $R\gtrless 1/\beta$. Therefore the unique equilibrium gross risk-free rate is $R=1/\beta$. Then the dynamics of wealth \eqref{eq:wdynamics} becomes $w_{t+1}=w_t$, so individual wealth is constant over time and the wealth distribution of each type is degenerate.

Finally, suppose $\set{\beta_j}_{j=1}^J$ take at least two distinct values. Let $G_j=(\beta_j R)^{1/\gamma_j}$ be the gross growth rate of wealth in \eqref{eq:wdynamics}. Since the numerator in \eqref{eq:clear} has the same sign as $G_j-1$ and $G_j\gtrless 1$ according as $\beta_j\gtrless 1/R$, the fact that $f(R)=0$ and $\set{\beta_j}_{j=1}^J$ are not all equal implies that there is some $j$ with $G_j>1$ and others with $G_j<1$.

An agent type with $G_j\le 1$ does not affect the upper tail of the wealth distribution since the wealth does not grow. Consider any type $j$ with $G_j>1$. Because type $j$ agents die with probability $p_j$ every period, the probability that an agent survives at least $n$ periods is $(1-p_j)^n$. Then wealth is at least $G_j^nw_{j0}$. Therefore we have
$$\Pr(w\ge G_j^nw_{j0})=(1-p_j)^n,$$
so letting $x=G_j^nw_{j0}$, we obtain $\Pr(w\ge x)=(x/w_{j0})^{-\alpha_j}$ for
$$\alpha_j=-\frac{\log(1-p_j)}{\log G_j}=-\gamma_j\frac{\log(1-p_j)}{\log (\beta_jR)}>0.$$
Therefore the wealth distribution of type $j$ has a Pareto upper tail with exponent $\alpha_j$. Noting that $(1-p_j)(\beta_jR)^{1/\gamma_j}<1$ and $\beta_jR>1$, taking the logarithm we obtain $\alpha_j>1$. Since the Pareto exponent of the entire cross-sectional distribution is the smallest exponent among all types, we obtain \eqref{eq:zeta}. \qedsymbol

\section{Policy iteration in income fluctuation problem}\label{sec:politer}

In this appendix we adopt the arguments in \cite*{LiStachurski2014} to characterize the solution to the income fluctuation problem by policy function iteration.\footnote{Our discussion is slightly different from \cite*{LiStachurski2014} due to the timing convention. In \cite*{LiStachurski2014}, $a$ is savings (end-of-period asset holdings) and the budget constraint is $c+a'=Ra+y$. In our framework, $a$ is beginning-of-period wealth and the budget constraint is $a'=R(a-c)+y'$. This change in the notation allows us to weaken the monotonicity requirement in Assumption 2.3 of \cite*{LiStachurski2014}.} Throughout this appendix, we maintain Assumptions \ref{asmp:utility} and \ref{asmp:boundedincome} but not Assumption \ref{asmp:BRRA}.

Let $S=\R_{++}\times Z$ be the state space, where $Z$ is as in Assumption \ref{asmp:boundedincome}. To identify a solution, let $\mathcal{C}$ be the set of functions $c:S\to\R$ such that $c(a,z)$ is increasing in $a$, $0<c(a,z)\le a$, and $\norm{u'\circ c-\phi}<\infty$ (sup norm), where $\phi(a,z)\equiv u'(a)$. The set $\mathcal{C}$ identifies a set of candidate consumption functions. For $c,d\in \mathcal{C}$, define the distance
$$\rho(c,d)=\norm{u'\circ c-u'\circ d}.$$

\begin{lem}\label{lem:metric}
$(\mathcal{C},\rho)$ is a complete metric space.
\end{lem}
\begin{proof}
Take any $c,d,e\in \mathcal{C}$. Clearly $\rho(c,d)\ge 0$, $\rho(c,d)=\rho(d,c)$, and
$$\rho(c,d)=0\iff (\forall a,z)u'(c(a,z))=u'(d(a,z))\iff c=d.$$
Furthermore, by the triangle inequality for the sup norm, we have
$$\rho(c,d)=\norm{u'\circ c-u'\circ d}\le \norm{u'\circ c-\phi}+\norm{u'\circ c-\phi}<\infty$$
and
$$\rho(c,d)=\norm{u'\circ c-u'\circ d}\le \norm{u'\circ c-u'\circ e}+\norm{u'\circ e-u'\circ d}=\rho(c,e)+\rho(e,d).$$
Therefore $(\mathcal{C},\rho)$ is a metric space.

To show completeness, suppose $\set{c_n}$ is a Cauchy sequence in $(\mathcal{C},\rho)$. Then $\set{u'(c_n(a,z))}$ is a Cauchy sequence in $\R$, so it has a limit $\mu$. Since $c_n(a,z)\le a$ and $u'$ is strictly decreasing, we have $u'(c_n(a,z))\ge u'(a)$ and hence $\mu\ge u'(a)$. Since $u'$ is continuous, strictly decreasing, and $\mu<\infty$, there exists a unique $c(a,z)\in (0,a]$ such that $u'(c(a,z))=\mu$. Since $u'$ is continuous, we have $c_n(a,z)\to c(a,z)$ pointwise. Since $c_n(a,z)$ is increasing in $a$, so is $c(a,z)$. Therefore $(\mathcal{C},\rho)$ is complete.
\end{proof}

For $c\in \mathcal{C}$, define $(Kc)(a,z)$ by the value $t\in (0,a]$ that satisfies the Euler equation \eqref{eq:coleman}.

\begin{lem}\label{lem:selfmap}
For any $c\in \mathcal{C}$ and $(a,z)\in S$, $(Kc)(a,z)$ is well-defined. Furthermore, $K:\mathcal{C}\to\mathcal{C}$ and $c\le d\implies Kc\le Kd$.
\end{lem}

\begin{proof}
Fix any $c\in \mathcal{C}$ and $(a,z)\in S$. For $t\in (0,a]$, define
$$g(t)=u'(t)-\max\set{\beta R\Ex{u'(c(R(a-t)+y',z'))}{z},u'(a)}.$$
The second term is finite because of the $\max$ operator and $u'(a)<\infty$. Therefore $g$ is finite on $(0,a]$. Since $u''<0$, $g$ is continuous and strictly decreasing. Furthermore, $g(0)=\infty$ and
$$g(a)=u'(a)-\max\set{\beta R\Ex{u'(c(y',z'))}{z},u'(a)}\le 0.$$
Therefore there exists a unique $t\in (0,a]$ such that $g(t)=0$, so \eqref{eq:coleman} holds.

To show $K:\mathcal{C}\to\mathcal{C}$, it remains to show that $(Kc)(a,z)$ is increasing in $a$. To show this, let $a_1\le a_2$ and $t_j=(Kc)(a_j,z)$ for $j=1,2$. Suppose on the contrary that $t_1>t_2$. Then using the fact that $c$ is increasing and $u'$ is strictly decreasing, we obtain
\begin{align*}
u'(t_2)>u'(t_1)&=\max\set{\beta R\Ex{u'(c(R(a_1-t_1)+y',z'))}{z},u'(a_1)}\\
&\ge \max\set{\beta R\Ex{u'(c(R(a_2-t_2)+y',z'))}{z},u'(a_2)}=u'(t_2),
\end{align*}
which is a contradiction. Therefore $t_1\le t_2$. The proof of $c\le d\implies Kc\le Kd$ is similar.
\end{proof}

\begin{lem}\label{lem:contraction}
If $\beta R<1$, then $\rho(Kc,Kd)\le \beta R\rho(c,d)$ for all $c,d\in \mathcal{C}$.
\end{lem}
\begin{proof}
Similar to the proof of Lemma A.5 in \cite*{LiStachurski2014}.
\end{proof}

\bibliographystyle{plainnat}
\bibliography{reference}

\end{document}